\documentclass{article}

\usepackage{arxiv}

\usepackage[utf8]{inputenc}
\usepackage[T1]{fontenc}
\usepackage{hyperref}
\hypersetup{
  colorlinks=true,
  linkcolor=black,
  citecolor=black,
  urlcolor=blue
}
\usepackage{booktabs}
\usepackage{amsfonts}
\usepackage{nicefrac}
\usepackage{microtype}
\usepackage{graphicx}
\usepackage{amsmath,amssymb}
\usepackage{bm}
\usepackage{mathrsfs}
\usepackage{amsthm}
\usepackage{multirow}
\usepackage{float}
\usepackage{xcolor}
\usepackage{algorithm}
\usepackage[noend]{algpseudocode}

\theoremstyle{plain}
\newtheorem{theorem}{Theorem}

\theoremstyle{definition}

\newcommand{\PNS}{PNS}
\newcommand{\Sop}{S_\theta}
\newcommand{\Sloop}{S_0}
\newcommand{\Rop}{R_\theta}
\newcommand{\Frame}{F_i}
\newcommand{\Gate}{\gamma_i}
\newcommand{\Net}{\eta_\theta}
\newcommand{\Feat}{\varphi_i}
\newcommand{\hloc}{h_i}

\title{Proximity-Preserving Neural Subdivision }

\author{
  Hassan Ugail \\
  Centre for Visual Computing and Intelligent Systems \\
  University of Bradford \\
  United Kingdom \\
   \\
}

\begin{document}
\maketitle

\begin{abstract}
Classical subdivision schemes are widely used because they are local, repeatable, and analytically tractable. A single stencil defines the entire refinement rule, and the behaviour of the resulting operator under iteration is well understood. This uniformity, however, means that fixed stencils tend to underfit localised geometric features, such as sharp ridges or soft edges, where curvature is concentrated. Neural mesh refinement can adapt to such features, yet unconstrained vertex prediction usually lacks the structural behaviour required of a subdivision operator once the refinement rule is applied to its own output. In this work, we introduce Proximity-Preserving Neural Subdivision, or PNS for short. PNS is a trainable refinement rule that augments Loop subdivision with a small, bounded, curvature-gated correction expressed in a covariant local frame. The construction is designed so that, for any finite network weights, the operator is exactly equivariant under rigid motion, reproduces planar input exactly, and remains inside a quadratic proximity envelope around the Loop stencil. At planar valence-k stars, the linearised operator agrees with Loop, and it therefore inherits Loop's tangent eigenspaces and Reif spectral gap at that reference configuration. All of these properties are architectural and hold before any training takes place. Empirically, PNS improves the approximation of localised ridge features while remaining inside its prescribed proximity envelope under repeated subdivision. An unconstrained neural baseline, in contrast, achieves stronger one-step fitting but develops high-frequency artefacts and leaves the subdivision regime once iterated. The overall message of this work is that learning can be introduced into subdivision without abandoning the structural constraints that make subdivision useful as a geometry-processing primitive.
\end{abstract}


\section{Introduction}
\label{sec:intro}

Subdivision surfaces are useful not only because they refine meshes, but because they define repeatable geometric operators. A subdivision rule can be applied many times to its own output, its behaviour under iteration can be analysed, and its limit is controlled by the structure of the stencil. This operator property is precisely what many neural mesh refinement methods lack. A network trained to predict refined vertex positions may produce accurate one-step output, yet repeated application can amplify local errors, introduce high-frequency artefacts, or drift out of the classical subdivision regime altogether. The gap between one-step accuracy and iterated behaviour is the phenomenon that this paper addresses.

The Loop scheme \cite{Loop1987} for triangle meshes and the Catmull--Clark scheme \cite{CatmullClark1978} for arbitrary-topology quad meshes have defined what is meant by smooth refinement for over four decades, and their reliability rests on the spectral theory of subdivision matrices \cite{DooSabin1978,Reif1995,PetersReif2008}. The strength of this classical setting is also its principal limitation. Stencils are fixed in advance, so the same local rule is applied across an entire mesh regardless of local geometric content. Feature-aware surface modelling has long recognised the problem. Piecewise-smooth reconstruction \cite{Hoppe1994} responded to this problem by introducing explicit creases at the cost of leaving the smooth subdivision setting. Curvature-adaptive and feature-sensitive schemes \cite{ZorinSchroederSweldens1996,Kobbelt2000,VlachosPetersBoydMitchell2001,BoubekeurAlexa2008} address part of the gap with hand-designed alternative stencils, yet remain limited to the behaviours their authors anticipated. Neural mesh refinement \cite{LiuKimChaudhuriAigermanJacobson2020,ChenKimAigermanJacobson2023,HanockaHertzFishGiryesFleishmanCohenOr2019,HuLiuGuoCaiHuangMuMartin2022,PotamiasPloumpisZafeiriou2022} provides a different response, replacing hand-designed adaptivity by learned adaptivity, typically at the cost of the structural behaviour required for repeated application.

The method proposed in this paper, called Proximity-Preserving Neural Subdivision, formulates learned refinement as a small, structurally restricted perturbation of Loop subdivision. For each interior edge with Loop-rule edge-vertex $q_i^0 = \tfrac{3}{8}(p_a + p_b) + \tfrac{1}{8}(p_c + p_d)$, the inserted vertex is
\begin{equation}
q_i^\theta \;=\; q_i^0 \;+\; h_i^2\, \Gate\, \Frame\, \Net(\Feat),
\label{eq:operator}
\end{equation}
in which $h_i$ is the local edge length. The factor $h_i^2$ forces second-order proximity to Loop. The scalar $\Gate \in [0,1]$ is a curvature gate that suppresses corrections on planar regions. The matrix $\Frame \in \mathrm{SO}(3)$ is a covariant local frame that yields exact rigid-motion equivariance. The network $\Net$ maps intrinsic features to $\mathbb{R}^3$ with a hard bound $\|\Net(\Feat)\| \le C$.

We do not claim that \PNS{} is the most expressive one-step surface fitter. An unconstrained neural predictor can often reduce one-step geometric error more aggressively, and it does so in our experiments. We claim instead that \PNS{} is a trainable subdivision operator. It improves finite-level feature approximation while preserving rigid-motion equivariance, affine reproduction, prescribed $O(h^2)$ proximity to Loop, planar spectral inheritance, and stable repeated application. The relevant comparison is not only one-step accuracy, but whether refinement remains controlled when the rule is iterated.

The contribution of this paper has four parts. We frame learned triangle-mesh refinement as a bounded, curvature-gated perturbation of Loop subdivision rather than as direct vertex prediction. We prove that this construction is $\mathrm{SE}(3)$-equivariant, affine-reproducing, and $O(h^2)$-proximate to Loop for any finite network weights, and that at planar valence-$k$ stars its linearisation coincides with Loop's. We verify these claims numerically, stress-test them adversarially, and ablate each architectural constraint to identify which property it protects. On a localised feature benchmark, \PNS{} improves finite-level approximation while remaining inside its proximity envelope under repeated subdivision, whereas an unconstrained neural baseline improves one-step fit but fails as an iterated operator. Figure~\ref{fig:hero} previews this behaviour. The position of the work is captured in one line. Classical subdivision provides reliability without adaptivity, unconstrained neural refinement provides adaptivity without reliability, and \PNS{} provides adaptive refinement while preserving the structural regime of classical subdivision.

\begin{figure}[t]
\centering
\includegraphics[width=0.94\textwidth]{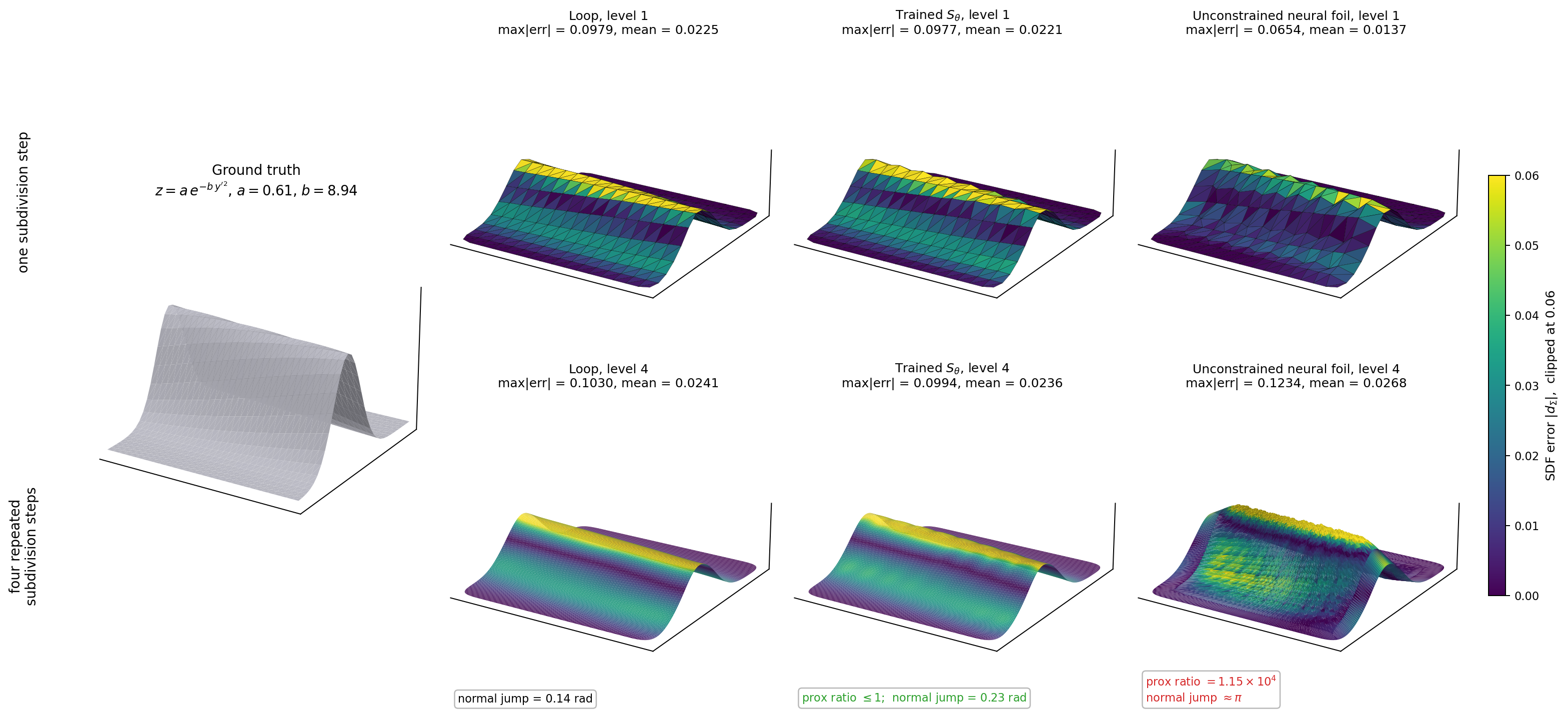}
\caption{Repeated refinement on a held-out ridge instance with parameters $a = 0.61$ and $b = 8.94$. The top row shows one subdivision step, and the bottom row shows four repeated steps. Annotation boxes report the final-level proximity ratio and the maximum face-normal jump. The unconstrained neural baseline fits aggressively at level one. By level four, however, its proximity ratio reaches $1.15 \times 10^4$, its face-normal jumps approach $\pi$, and high-frequency texture artefacts are visible. \PNS{} remains inside the architectural cap with proximity ratio at most unity, producing a stable, feature-sensitive surface across all four levels with normal jump comparable to Loop's.}
\label{fig:hero}
\end{figure}

\section{Related work}
\label{sec:related}

\subsection{Classical subdivision and adjacent constructive traditions}

The canonical approximating subdivision schemes are Loop \cite{Loop1987} for triangle meshes and Catmull--Clark \cite{CatmullClark1978} for arbitrary-topology quad meshes. Doo and Sabin \cite{DooSabin1978} introduced an early scheme together with the spectral analysis at extraordinary points that underlies subsequent theory. Modified Butterfly \cite{ZorinSchroederSweldens1996} extends the four-point interpolatory curve scheme of Dyn, Gregory, and Levin \cite{DynGregoryLevin1987} to triangle meshes. The $\sqrt{3}$-subdivision scheme of Kobbelt \cite{Kobbelt2000} offers a different refinement topology. PN-triangles \cite{VlachosPetersBoydMitchell2001} and Phong tessellation \cite{BoubekeurAlexa2008} introduce normal-driven local rules. Taubin's curvature-flow smoothing without shrinkage \cite{Taubin1995} is a canonical analytic mesh-modification rule. Exact evaluation of Catmull--Clark surfaces is due to Stam \cite{Stam1998}. Textbook treatments \cite{WarrenWeimer2002,PetersReif2008,PrautzschBoehmPaluszny2002} consolidate the theory, and the review of Dyn and Levin \cite{DynLevin2002} surveys convergence and smoothness analysis. Curvature-adaptive variants mix fixed schemes according to hand-tuned criteria, which is a discrete switch rather than a learned correction.

A complementary constructive paradigm is the partial-differential-equation surface tradition, which produces controlled limiting geometry by solving a boundary-value problem rather than by iterating a local mask. It was introduced by Bloor and Wilson \cite{BloorWilson1989,BloorWilson1990} and later extended to interactive design, aircraft geometry, patchwise mesh approximation, and facial parameterisation \cite{UgailBloorWilson1999a,UgailBloorWilson1999b,KubiesaUgailWilson2004,AthanasopoulosUgailCastro2009,ShengSourinCastroUgail2010,ShengWillisCastroUgail2011}. Related work on harmonic and biharmonic B\'ezier surfaces \cite{MonterdeUgail2004,MonterdeUgail2006} and a survey of the field \cite{GonzalezCastroUgailWillisPalmer2008} provide additional context. The tension between local adaptivity and structural safety arises in that tradition through a different construction, and the present work stays inside the subdivision tradition and asks how a local operator can be learned while retaining the structural properties that make subdivision analysable.

\subsection{Proximity analysis of nonlinear subdivision}

Proximity analysis, due to Wallner and Dyn \cite{WallnerDyn2005}, provides a general bridge between linear and nonlinear or projection-based subdivision schemes. If a scheme $T$ is $O(h^2)$-proximate to a linear scheme $S$, then $T$ inherits the regular-region $C^1$ convergence of $S$ under standard hypotheses. Grohs \cite{Grohs2009,Grohs2010} extended these ideas to manifold-valued subdivision and to the stability of multiscale transformations. H\"uning and Wallner \cite{HueningWallner2019,HueningWallner2022} treat convergence on Riemannian manifolds. A recent Heisenberg-group four-point construction \cite{UgailCarriazo2026} shows how sharply the proximity hypothesis matters. Its group-law correction is only linearly close to a Euclidean reference, and the central-channel limit is Zygmund but not $C^1$. Departing from the proximity regime by even one order can therefore destroy limit smoothness. Our analysis follows this proximity viewpoint, and our contribution is architectural. The network is parameterised so that the proximity hypothesis used by Wallner and Dyn is satisfied by construction, and the bounded $h_i^2$-scaled correction in Equation~\eqref{eq:operator} guarantees $O(h^2)$ proximity for any network weights.

\subsection{Neural mesh refinement and equivariant geometric learning}

The most direct point of comparison is Neural Subdivision \cite{LiuKimChaudhuriAigermanJacobson2020}, which trains a network to predict refined vertex positions on a Loop-style refined topology. Neural Progressive Meshes \cite{ChenKimAigermanJacobson2023} learns coarse-to-fine mesh representations, building on the progressive-meshes line of work \cite{Hoppe1996}. SubdivNet \cite{HuLiuGuoCaiHuangMuMartin2022} uses Loop-subdivision sequence connectivity to define a mesh convolutional network. MeshCNN \cite{HanockaHertzFishGiryesFleishmanCohenOr2019} introduces an edge-based mesh convolution, and Neural Mesh Simplification \cite{PotamiasPloumpisZafeiriou2022} learns coarsening rather than refinement.

Our work differs from these in a fundamental way. Prior neural refinement methods predict vertex positions or unbounded vertex offsets freely. They can fit a target surface aggressively in one step, and yet they are not generally equivariant by construction, do not reproduce affine input, and have no fixed proximity envelope to a classical scheme. Methods such as SubdivNet use subdivision connectivity for representation learning rather than learning a subdivision operator. The present work is different in kind. We learn a parametrised member of a constrained operator class whose structural behaviour is fixed before training. Table~\ref{tab:distinction} states this distinction concisely. Neural Subdivision merits a more detailed contrast because it is the closest match in aim. It predicts a vertex displacement directly from a learned mesh representation, with no architectural bound on displacement magnitude, no built-in equivariance under rigid motion, and no built-in proximity to a classical Loop stencil. \PNS{} differs in three ways. The learned correction in Equation~\eqref{eq:operator} is bounded above by $C h_i^2$ for any finite network weights, so \PNS{} is $O(h^2)$-proximate to Loop by construction. The correction is expressed in a covariant local frame and depends only on rigid-motion-invariant features, so the operator is exactly $\mathrm{SE}(3)$-equivariant for any weights. The curvature gate vanishes on planar input together with its first derivative, so \PNS{} reproduces Loop pointwise on planar input and inherits Loop's linearisation at planar valence-$k$ stars.

\begin{table}[t]
\centering
\caption{The distinction in kind between \PNS{} and prior neural refinement. \PNS{} learns a parametrised member of a constrained operator class rather than an unconstrained mapping.}
\label{tab:distinction}
\begin{tabular*}{\textwidth}{@{\extracolsep{\fill}}p{0.44\textwidth} p{0.48\textwidth}@{}}
\toprule
\textbf{Prior neural refinement} & \textbf{\PNS{} (this work)}\\
\midrule
predicts vertex positions or offsets freely & predicts bounded, gated corrections to a classical stencil\\
optimised for one-shot fitting & designed to remain inside a proximity envelope across iterations\\
generally not $\mathrm{SE}(3)$-equivariant by construction & exactly $\mathrm{SE}(3)$-equivariant by construction\\
no proximity envelope to a classical scheme & $O(h^2)$-proximate to Loop for any weights\\
no spectral inheritance & exact planar spectral inheritance at valence-$k$ stars\\
neural upsampler & trainable subdivision operator\\
\bottomrule
\end{tabular*}
\end{table}

Intrinsic features and covariant local frames also have a substantial history in geometric deep learning. Geodesic convolutional networks \cite{MasciBoscainiBronsteinVandergheynst2015} introduced patch-based intrinsic learning on manifolds. Gauge-equivariant convolutional networks \cite{CohenWeilerKicanaogluWelling2019} and gauge-equivariant mesh convolutional networks \cite{deHaanWeilerCohenWelling2021} make the local-frame structure explicit and design networks that are equivariant under frame changes. SplineCNN \cite{FeyLenssenWeichertMueller2018}, PointNet and PointNet++ \cite{QiSuMoGuibas2017,QiYiSuGuibas2017}, DGCNN \cite{WangSunLiuSarmaBronsteinSolomon2019}, and ChebNet \cite{DefferrardBressonVandergheynst2016} provide further baselines and formulations. Our use of equivariance is narrower and more direct than in these works. We do not learn equivariance. We hard-code rigid-motion equivariance by computing intrinsic features and expressing the network output in a covariant local frame. The equivariance error of the resulting operator is at machine precision for any weights, without any equivariance loss term.

\section{Method}
\label{sec:method}

\subsection{Loop background and proximity}
\label{sec:loop}

We work with oriented manifold triangle meshes $M = (V, E, F)$ in the standard polygon-mesh-processing setting of Botsch et al.\ \cite{BotschKobbeltPaulyAlliezLevy2010}. Vertex positions are denoted $p_i \in \mathbb{R}^3$. The mesh size is $h(M) = \max_{e \in E}\|e\|$, and the local edge length at edge $e_i = (v_a, v_b)$ at level $\ell$ is the pre-refinement Euclidean length $\hloc = \|p_b^\ell - p_a^\ell\|$. All proximity ratios at level $\ell$ are computed relative to this length. For an interior edge $e_i$ with opposite vertices $v_c, v_d$, Loop subdivision \cite{Loop1987} inserts the new edge vertex
\begin{equation}
q_i^0 \;=\; \tfrac{3}{8}(p_a + p_b) + \tfrac{1}{8}(p_c + p_d),
\label{eq:loop-edge}
\end{equation}
and updates an old vertex $p_j$ with $1$-ring neighbours $\mathcal{N}(j)$ and valence $k$ to the smoothed position
\begin{equation}
p_j^0 \;=\; (1 - k\beta_k)\, p_j \;+\; \beta_k \sum_{\ell \in \mathcal{N}(j)} p_\ell, \qquad
\beta_k \;=\; \tfrac{1}{k}\!\left[\tfrac{5}{8} - \!\left(\tfrac{3}{8} + \tfrac{1}{4}\cos\tfrac{2\pi}{k}\right)^{\!2}\right].
\label{eq:loop-vertex}
\end{equation}
Loop subdivision is affine-invariant, $C^2$ in regular regions, and at most $C^1$ at extraordinary vertices \cite{PetersReif2008}. A refinement operator $T$ acting on the same connectivity as $S$ is $O(h^2)$-proximate to $S$ if there is a constant $C \ge 0$, independent of the mesh, such that $\|T(M) - S(M)\|_\infty \le C h(M)^2$ on every shape-regular mesh in the considered class. Proximity is the standard tool by which nonlinear or projection-based subdivision schemes inherit convergence from a linear reference scheme \cite{WallnerDyn2005,Grohs2009}. Near a valence-$k$ vertex, one step of Loop subdivision is a linear map on a $(k{+}1)$-dimensional space of star configurations. Reif's conditions \cite{Reif1995} require a single subdominant eigenvalue, a complex-conjugate pair of tangent eigenvalues $\lambda_t = |\lambda_t|\,e^{\pm i\,2\pi/k}$, and a strictly larger gap $|\lambda_t| - |\lambda_3|$ to subsequent eigenvalues. For Loop, $|\lambda_t| = \tfrac{3}{8} + \tfrac{1}{4}\cos(2\pi/k)$.

\subsection{PNS operator and algorithm}
\label{sec:components}

The operator $\Sop$ acts on a triangle mesh $M$ by computing a refined topology identical to Loop's $1$-to-$4$ split, applying Loop's old-vertex update from Equation~\eqref{eq:loop-vertex} to existing vertices, and inserting a learned edge vertex at every interior edge. Written as an additive perturbation of Loop, the operator takes the form
\begin{equation}
\Sop \;=\; \Sloop + \Rop, \qquad q_i^\theta \;=\; q_i^0 \;+\; \hloc^2\, \Gate\, \Frame\, \Net(\Feat).
\label{eq:operator-restate}
\end{equation}
Boundary edges receive Loop's standard boundary midpoint, so the boundary is not modified by the network. The correction $\Rop$ is built from four components, each carrying a specific structural role.

The intrinsic feature vector $\Feat = (\Feat^{\mathrm{curv}}, \Feat^{\mathrm{shape}})$ encodes the local geometry visible from edge $e_i$ in a form invariant under rigid motion. The curvature component
\begin{equation*}
\Feat^{\mathrm{curv}} = 1 - n_{f_1}\!\cdot n_{f_2}
\end{equation*}
drives the gate and vanishes exactly on planar configurations, in the spirit of the discrete differential-geometry operators of Meyer et al.\ \cite{Meyer2003,Desbrun2000}. The shape component is the four-tuple of intrinsic edge-length ratios
\begin{equation*}
\Feat^{\mathrm{shape}} = \!\Big(\log\tfrac{\|p_a-p_c\|}{\hloc},\, \log\tfrac{\|p_b-p_c\|}{\hloc},\, \log\tfrac{\|p_a-p_d\|}{\hloc},\, \log\tfrac{\|p_b-p_d\|}{\hloc}\Big).
\end{equation*}
Both feature groups are invariant under rigid motion. The curvature gate is a soft smooth indicator of curvature concentration,
\begin{equation}
\Gate \;=\; \tanh\!\big(c\,\|\Feat^{\mathrm{curv}}\|^2\big), \qquad c > 0.
\label{eq:gate}
\end{equation}
It is bounded in $[0,1]$, and both $\Gate$ and its gradient vanish at $\Feat^{\mathrm{curv}} = 0$, because the argument of $\tanh$ has a zero of order two there. The gate therefore suppresses corrections in regions where Loop should be preserved, and the linearisation of $\Sop$ at a planar valence-$k$ star coincides exactly with Loop's.

For each interior edge we construct a covariant local frame $\Frame = [T_i,\,N_i,\,B_i] \in \mathrm{SO}(3)$ from the unit edge vector $T_i = (p_b - p_a)/\|p_b - p_a\|$, the projected average of adjacent face normals $N_i = (n_{f_1} + n_{f_2} - ((n_{f_1} + n_{f_2})\cdot T_i)T_i)/\|\cdot\|$, and $B_i = T_i \times N_i$. This frame transforms covariantly under rigid motion, and combined with rigid-motion-invariant features it makes the displacement $\Frame \Net(\Feat)$ equivariant by construction. The frame is smooth on shape-regular stencils that satisfy three explicit conditions. The edge must have positive length. The two adjacent faces must not be antiparallel, so that $n_{f_1} \cdot n_{f_2} > -1 + \epsilon$ for a fixed small $\epsilon > 0$. The angle between $n_{f_1} + n_{f_2}$ and $T_i$ must be bounded away from zero, which is automatic on shape-regular stencils with bounded aspect ratio. Near-degenerate frames are handled by a deterministic fallback direction drawn from a fixed world-coordinate triad.

The learned correction is a small multilayer perceptron whose output is hard-bounded through $L_2$ normalisation,
\begin{equation*}
\Net(\Feat) \;=\; C \cdot \frac{z_\theta(\Feat)}{\max\!\big(1,\, \|z_\theta(\Feat)\|\big)}, \qquad \|\Net(\Feat)\| \le C.
\end{equation*}
The bound is enforced by construction, independent of network weights and training dynamics. The single hyperparameter $C$ controls the per-edge correction budget and, as shown in Section~\ref{sec:theory}, the spectral safety margin at valence-$k$ stars. Because the correction is scaled by $\hloc^2$, its magnitude vanishes under repeated refinement on shape-regular sequences. \PNS{} is therefore a controlled finite-level adaptive refinement rule rather than a free neural limit-surface generator, with learning used to spend a per-level $Ch_i^2$ budget where a hand-designed stencil would underfit. Algorithm~\ref{alg:pns} summarises one step. Each edge is processed independently and the network evaluation batches naturally across edges.

\begin{algorithm}[t]
\caption{One step of \PNS{} subdivision.}
\label{alg:pns}
\begin{algorithmic}[1]
\Require mesh $M = (V, F)$, budget $C$, network weights $\theta$
\For{each old vertex $p_j$}
  \State $p_j' \gets$ Loop old-vertex update from Equation~\eqref{eq:loop-vertex}
\EndFor
\For{each interior edge $e_i = (v_a, v_b)$}
  \State $q_i^0 \gets \tfrac{3}{8}(p_a + p_b) + \tfrac{1}{8}(p_c + p_d)$
  \State $\Feat \gets$ intrinsic curvature and shape features
  \State $\Gate \gets \tanh\!\big(c \|\Feat^{\mathrm{curv}}\|^2\big)$
  \State $\Frame \gets$ covariant frame at $e_i$
  \State $z \gets z_\theta(\Feat), \quad \Net \gets C \cdot z / \max(1, \|z\|)$
  \State $q_i^\theta \gets q_i^0 + \hloc^2 \Gate \Frame \Net$
\EndFor
\For{each boundary edge}
  \State apply Loop boundary midpoint
\EndFor
\State apply $1$-to-$4$ refinement to $F$
\State \Return refined mesh $M'$
\end{algorithmic}
\end{algorithm}

\section{Architectural guarantees}
\label{sec:theory}

The properties below hold for the operator class defined by Equation~\eqref{eq:operator-restate} for any finite network weights, and are not consequences of training. We work under standing assumptions that meshes are oriented manifolds with shape-regular stencils, that interior edges admit a well-defined covariant frame, and that $C$ is finite. The results are intentionally simple consequences of the parameterisation. Their role is not to introduce new subdivision theory but to certify that learning cannot violate the structural constraints listed in Table~\ref{tab:structural}, which pairs each guarantee with its architectural mechanism, its theorem, and its empirical certification in Section~\ref{sec:certification}.

\begin{table}[t]
\centering
\caption{Structural guarantees and how they are realised. Each row pairs an architectural mechanism with the property it secures, the theorem that establishes the property, and the numerical check that certifies it.}
\label{tab:structural}
\begin{tabular*}{\textwidth}{@{\extracolsep{\fill}}p{0.24\textwidth} p{0.34\textwidth} c p{0.22\textwidth}@{}}
\toprule
\textbf{Property} & \textbf{Mechanism} & \textbf{Theorem} & \textbf{Empirical check}\\
\midrule
bounded displacement & $\|\Net\| \le C$, $\Gate \le 1$ & T\ref{thm:bounded} & max ratio\\
$\mathrm{SE}(3)$-equivariance & intrinsic features and covariant frame & T\ref{thm:equivariance} & machine-precision test\\
affine reproduction & $\Gate = 0$ at $\Feat^{\mathrm{curv}} = 0$ & T\ref{thm:affine} & exact-zero test\\
$O(h^2)$ proximity & $\hloc^2$ scaling & T\ref{thm:proximity} & slope test\\
regular-region $C^1$ & proximity and WD'05 & T\ref{thm:regular} & inheritance\\
$O(h)$ normal variation & cross-product sensitivity & T\ref{thm:normal-variation} & slope test\\
local Lipschitz stability & each factor Lipschitz & T\ref{thm:lipschitz} & slope test\\
planar spectral inheritance & $\Gate(0)=0$, $\nabla\Gate(0)=0$ & T\ref{thm:planar-spectral} & spectral test\\
\bottomrule
\end{tabular*}
\end{table}

\subsection{One-step properties}

The operator has three properties that hold pointwise on a single subdivision step. Bounded displacement caps the per-edge correction at the architectural budget. Equivariance means that the operator commutes with rigid motion. Affine reproduction means that the operator coincides with Loop on planar input.

\begin{theorem}[Bounded displacement]
\label{thm:bounded}
For every interior edge,
$\|q_i^\theta - q_i^0\| \le C\, \hloc^2$.
\end{theorem}

\begin{proof}
The displacement vector is $\hloc^2 \Gate \Frame \Net(\Feat)$. By construction $\|\Net(\Feat)\| \le C$, since the network output is normalised to lie in a ball of radius $C$. The gate satisfies $\Gate \in [0,1]$, because $\tanh$ takes values in $(-1,1)$ and its argument is non-negative. The frame is an element of $\mathrm{SO}(3)$ and is therefore an isometry of $\mathbb{R}^3$. Combining these three facts yields $\|\hloc^2 \Gate \Frame \Net(\Feat)\| \le \hloc^2 \cdot 1 \cdot 1 \cdot C = C \hloc^2$.
\end{proof}

A practical corollary is that on shape-regular stencils, with $\hloc$ small relative to the local altitude, no incident triangle is inverted by the correction.

\begin{theorem}[Equivariance]
\label{thm:equivariance}
For every $g = (R, t) \in \mathrm{SE}(3)$, $\Sop(g M) = g\, \Sop(M)$.
\end{theorem}

\begin{proof}
Loop subdivision is affine-invariant, so $q_i^0(gM) = R q_i^0(M) + t$. The intrinsic features $\Feat$ are functions of squared edge lengths and dot products of unit face normals, which are invariant under any rigid transformation. Hence $\Feat(gM) = \Feat(M)$, so $\Net(\Feat(gM)) = \Net(\Feat(M))$ and $\Gate(gM) = \Gate(M)$. The local frame transforms covariantly, with $\Frame(gM) = R \Frame(M)$, since the tangent direction and the projected face-normal average transform by $R$ and the cross product is equivariant under $\mathrm{SO}(3)$. Combining these,
\begin{align*}
q_i^\theta(gM) &= q_i^0(gM) + \hloc^2(gM)\, \Gate(gM)\, \Frame(gM)\, \Net(\Feat(gM)) \\
&= R q_i^0(M) + t + \hloc^2(M)\, \Gate(M)\, R \Frame(M)\, \Net(\Feat(M)) \\
&= R \big(q_i^\theta(M)\big) + t.
\end{align*}
\end{proof}

\begin{theorem}[Affine reproduction]
\label{thm:affine}
If $M$ samples an affine planar surface, then $\Sop(M) = \Sloop(M)$.
\end{theorem}

\begin{proof}
On an affine planar mesh, oriented manifold connectivity gives $n_{f_1} = n_{f_2}$ for every interior edge, so $\Feat^{\mathrm{curv}} = 0$ and $\Gate = 0$. The correction term in Equation~\eqref{eq:operator-restate} vanishes identically.
\end{proof}

Without the gate, or with a gate that does not vanish at $\Feat^{\mathrm{curv}} = 0$, the network would in general perturb planar input.

\subsection{Proximity, first-order stability, and spectral inheritance}

The bounded-displacement bound translates directly into an asymptotic proximity bound to Loop and places $\Sop$ in the regime of the proximity-based convergence inheritance theorem of Wallner and Dyn \cite{WallnerDyn2005}. It also implies first-order stability of face normals and Lipschitz stability of the operator under vertex perturbations. Finally, the vanishing gradient of the gate at the planar reference gives an exact spectral inheritance result at planar extraordinary vertices.

\begin{theorem}[Proximity]
\label{thm:proximity}
For any mesh in the considered class, $\|\Sop(M) - \Sloop(M)\|_\infty \le C\, h(M)^2$.
\end{theorem}

\begin{proof}
Old vertices are unchanged under $\Rop$, and boundary edge midpoints are unchanged. For each inserted interior edge vertex, Theorem~\ref{thm:bounded} gives $\|q_i^\theta - q_i^0\| \le C \hloc^2 \le C h(M)^2$. Taking the supremum yields the stated bound.
\end{proof}

This is the central architectural guarantee. It is the counterpart of the proximity assumption used by Wallner and Dyn, and it is the property that fails for unconstrained neural refinement, since a free vertex predictor has no a priori bound.

\begin{theorem}[Regular-region proximity inheritance]
\label{thm:regular}
Consider a sequence of shape-regular valence-$6$ refinement neighbourhoods on which Loop has uniform $C^1$ convergence and on which the \PNS{} correction map is Lipschitz in the local stencil coordinates. Under the standing assumptions, the \PNS{} sequence is $O(h^2)$-proximate to Loop and inherits Loop's regular-region $C^1$ convergence behaviour.
\end{theorem}

\begin{proof}
Loop is $C^2$ in regular regions \cite{PetersReif2008}, and Theorem~\ref{thm:proximity} provides an $O(h^2)$ proximity bound. The gate is smooth in the mesh vertex array, the frame is smooth on shape-regular stencils away from frame degeneracies, and the network is smooth as a multilayer perceptron with smooth activations. The perturbation field $\Rop = \Sop - \Sloop$ therefore has the regularity required by the proximity inheritance theorem of Wallner and Dyn \cite[Section 3]{WallnerDyn2005}. Applying that theorem with reference $\Sloop$ and perturbation $\Rop$ yields the claim.
\end{proof}

We do not claim global $C^1$ convergence on arbitrary topology. Theorem~\ref{thm:regular} is a local regular-region statement.

\begin{theorem}[Normal variation]
\label{thm:normal-variation}
On shape-regular stencils, the discrete face normals after one step of $\Sop$ satisfy $\| n_\theta - n_0 \| = O(h)$, where $n_0$ are the corresponding Loop normals.
\end{theorem}

\begin{proof}
Let $u, v$ denote two edges of a face after applying $\Sloop$, and let $\delta_u, \delta_v$ denote the perturbations of these edges introduced by $\Rop$. Theorem~\ref{thm:proximity} bounds vertex displacements by $C h(M)^2$, so $\|\delta_u\| = O(h^2)$ and $\|\delta_v\| = O(h^2)$. The bilinear identity
\begin{equation*}
(u + \delta_u) \times (v + \delta_v) - u \times v \;=\; u \times \delta_v + \delta_u \times v + \delta_u \times \delta_v
\end{equation*}
has right-hand side of norm $O(h^3)$, since $\|u\|, \|v\| = O(h)$ on shape-regular stencils. The Loop cross product has norm $O(h^2)$, equal to twice the Loop face area. Dividing gives a unit-normal perturbation of order $h^3/h^2 = h$.
\end{proof}

We make no claim of $O(h^2)$ normal variation or curvature continuity. The bound is per-level and relative to the Loop-refined mesh at the same level. It does not by itself bound accumulated normal error over many iterations, which we evaluate empirically in Section~\ref{sec:experiments}.

\begin{theorem}[Lipschitz stability]
\label{thm:lipschitz}
If the feature map $\Feat$, the frame map $\Frame$, and the network $\Net$ are Lipschitz on the shape-regular mesh class, then $\Sop$ is locally Lipschitz with respect to vertex perturbations.
\end{theorem}

\begin{proof}
Each ingredient is Lipschitz on the shape-regular class. Loop's stencil is a fixed linear map. The features are smooth functions of squared edge lengths and unit normals. The frame is a smooth function of the same quantities. The gate is the composition of a quadratic polynomial with $\tanh$, both Lipschitz on bounded sets. The network is Lipschitz by hypothesis. Multiplying bounded Lipschitz factors and adding the affine map $\Sloop$ yields a locally Lipschitz operator.
\end{proof}

\begin{theorem}[Planar spectral inheritance]
\label{thm:planar-spectral}
Let $M_{\mathrm{flat},k}$ denote a planar valence-$k$ star. Then $D\Sop(M_{\mathrm{flat},k}) = D\Sloop(M_{\mathrm{flat},k})$, and the tangent eigenvalues $\lambda_t = |\lambda_t|\,e^{\pm i\,2\pi/k}$ with $|\lambda_t| = \tfrac{3}{8} + \tfrac{1}{4}\cos(2\pi/k)$ and the Reif spectral gap of $\Sop$ at $M_{\mathrm{flat},k}$ coincide with Loop's.
\end{theorem}

\begin{proof}
At $M_{\mathrm{flat},k}$ the curvature features vanish, so $\Gate = 0$, and by Theorem~\ref{thm:affine} the correction vanishes pointwise. The correction has the form $\Rop = \hloc^2 \cdot \tanh(c x^2) \cdot \Frame \cdot \Net$, where $x = \|\Feat^{\mathrm{curv}}\|$ depends smoothly on the mesh vertex array. The Jacobian of $\tanh(c x^2)$ factors as $\mathrm{sech}^2(c x^2) \cdot 2 c x \cdot \nabla x$. At $x = 0$ the factor $2 c x$ vanishes, so the entire Jacobian vanishes. Since the gate is a multiplicative factor and has a zero of order two at the planar reference, both $\Rop$ and $D\Rop$ vanish at $M_{\mathrm{flat},k}$. The linearisations therefore agree, and the eigenvalue and Reif-gap claims follow immediately.
\end{proof}

Theorem~\ref{thm:planar-spectral} is a planar spectral inheritance result, not a global proof of $C^1$ continuity at extraordinary vertices. Off the planar reference, the spectrum of $\Sop$ depends on $C$ and on the network weights, and we provide numerical certificates in Section~\ref{sec:certification}.

\section{Numerical certification}
\label{sec:certification}

We treat the architectural guarantees of Section~\ref{sec:theory} as testable predictions and verify them numerically. Because the properties hold for any finite weights, the strongest form of certification is to seek violations adversarially, and we use both untrained networks and adversarially trained ones throughout.

\subsection{Structural certification and ablation}

For each property we quantify the empirical error on a battery of shape-regular paraboloid and bumpy-sphere meshes. Table~\ref{tab:structural-numerical} summarises the measurements. The equivariance error is at machine precision. The affine reproduction error is exactly zero, because the gate evaluates to zero at planar configurations. Refining the mesh and measuring the proximity-bound exponent recovers slope $2.000$ for proximity and slope $1.000$ for normal variation. The empirical Lipschitz slope is approximately $1.02$. Adversarial searches then attempt to violate the architectural bounds directly. Maximising the proximity ratio $\|q_i^\theta - q_i^0\| / (C \hloc^2)$ by gradient ascent on both the network weights and the mesh configuration, with $C$ held fixed, leaves the ratio at or below unity within numerical tolerance throughout the search. Maximising the equivariance error similarly leaves it at machine precision. The certification is by construction, and no choice of weights produces a mesh that violates Theorems~\ref{thm:bounded} or~\ref{thm:equivariance}.

\begin{table}[t]
\centering
\caption{Structural tests on shape-regular paraboloid and bumpy-sphere meshes. Each row reports the expected value and the empirical measurement.}
\label{tab:structural-numerical}
\begin{tabular*}{0.7\textwidth}{@{\extracolsep{\fill}}l c c@{}}
\toprule
\textbf{Test} & \textbf{Expected} & \textbf{Measured}\\
\midrule
Equivariance error & 0 & $3.5 \times 10^{-15}$\\
Affine reproduction error & 0 & $0$\\
Proximity scaling slope & 2 & $2.000$\\
Normal-variation slope & 1 & $1.000$\\
Stability slope & 1 & $\approx 1.02$\\
\bottomrule
\end{tabular*}
\end{table}

We complement this certification with an architectural ablation that removes each element in turn and measures which structural property breaks. Table~\ref{tab:ablations} summarises the pattern. Removing the $h^2$ scaling breaks asymptotic proximity. Removing the gate breaks affine reproduction. Removing boundedness breaks the proximity envelope. Replacing the entire structured operator with an unconstrained multilayer perceptron eliminates all architectural guarantees. Each architectural element protects a specific property. The unconstrained variant fits better in one step on every metric tested, consistent with the framing that one-step fitting and repeated-application stability are different problems. The table also includes a \PNS{}-with-random-$\theta$ row, in which only the network weights are replaced by random initialisation. Since the untrained network provides no informative direction inside the $Ch^2$ envelope, the one-step fit degrades to the level of the Loop reference, while the architectural columns remain unaffected because they are consequences of the parameterisation.

\begin{table}[t]
\centering
\caption{Architectural ablation, showing which structural property breaks when each constraint is removed. The random-$\theta$ row keeps the full architecture and replaces trained weights by random initialisation.}
\label{tab:ablations}
\begin{tabular*}{\textwidth}{@{\extracolsep{\fill}}l c c c c@{}}
\toprule
\textbf{Variant} & \textbf{Proximity} & \textbf{Affine repro.} & \textbf{Repeated stable} & \textbf{One-step fit}\\
\midrule
\PNS{} (full, trained) & yes & yes & yes & moderate\\
\PNS{} with random $\theta$ & yes & yes & yes & near Loop\\
no $h^2$ scaling & no asymptotic & yes & weak & better\\
no gate & yes (one step) & no & weak & better\\
unbounded output & no bound & maybe & weak & better\\
unconstrained MLP & no & no & fails under iter.\ & best\\
\bottomrule
\end{tabular*}
\end{table}

\subsection{Spectral certification}

We verified Loop's analytic tangent eigenvalues numerically for $k = 3, \ldots, 12$ by extracting the local subdivision matrix from a planar valence-$k$ star and diagonalising it. We also verified Theorem~\ref{thm:planar-spectral} by computing the local linearisation $D\Sop$ at $M_{\mathrm{flat},k}$ and confirming that it equals $D\Sloop$ to machine precision. Off the planar reference, we measured the Reif spectral gap $|\lambda_t| - |\lambda_3|$ as a function of perturbation magnitude $\delta$ and budget $C$. The gap is preserved at $\delta = 0$ for any $C$, and it degrades smoothly with $\delta$ at a rate proportional to $C$. Adversarial gradient ascent on a gap-collapse objective confirms that $C$ acts as a spectral safety parameter, in the sense that smaller $C$ widens the safety margin against spectral collapse.

\section{Experiments}
\label{sec:experiments}

\subsection{Setup, baselines, and metrics}
\label{sec:setup}

We focus the trained-operator evaluation on a synthetic family of Gaussian ridges $z(x, y) = a \exp(-b (y')^2)$, where $y'$ is the across-ridge coordinate after a random rotation. This family is well suited to isolating the regime in which bounded $O(h^2)$ corrections should help. Fixed Loop stencils smooth the ridge cap because the inserted edge vertex receives only one-eighth weight from the diagonal cross of opposite vertices. Truncation error concentrates near the ridge, where a finite correction budget is most useful. The analytic ground-truth surface gives exact signed-distance and analytic-normal references, so paired evaluation is unambiguous. The benchmark is not intended to demonstrate universal superiority over Loop. Section~\ref{sec:discussion} discusses a saddle null result on smooth surfaces where this regime does not apply. The dataset comprises $30$ random ridge instances with $a \in [0.3, 0.8]$ and $b \in [3, 10]$, together with random rotation and a small centre offset. It is split $20/10$ into training and held-out test sets.

The training loss combines geometric and structural terms,
\begin{equation*}
\mathcal{L} \;=\; \lambda_{\mathrm{sdf}}\mathcal{L}_{\mathrm{sdf}} + \lambda_{\mathrm{n}}\mathcal{L}_{\mathrm{n}} + \lambda_{\mathrm{f}}\mathcal{L}_{\mathrm{fair}} + \lambda_{\mathrm{p}}\mathcal{L}_{\mathrm{prox}}.
\end{equation*}
The signed-distance term $\mathcal{L}_{\mathrm{sdf}}$ penalises the squared analytic signed distance at inserted vertices. The normal term $\mathcal{L}_{\mathrm{n}}$ aligns the discrete inserted-vertex normal with the analytic surface normal. The fairness term $\mathcal{L}_{\mathrm{fair}}$ penalises high-frequency normal variation along the ridge. The proximity term $\mathcal{L}_{\mathrm{prox}}$ is a soft regulariser on the proximity ratio. This proximity penalty is not the source of the proximity bound, which is architectural. It merely discourages the network from spending the full envelope when a smaller correction suffices. In our experiments the loss weights are $\lambda_{\mathrm{sdf}} = 1$, $\lambda_{\mathrm{n}} = 0.1$, $\lambda_{\mathrm{f}} = 0.01$, and $\lambda_{\mathrm{p}} = 0.01$. The network $\Net$ is a three-layer multilayer perceptron with sixty-four hidden units per layer and $\tanh$ activations. Optimisation uses Adam \cite{KingmaBa2015} with a fixed learning rate of $10^{-3}$, batch size sixteen, and one hundred and eighty training steps.

We compare against seven other methods that share the same evaluation harness. Loop subdivision \cite{Loop1987} serves as the reference scheme. Modified Butterfly \cite{ZorinSchroederSweldens1996} and PN-triangles \cite{VlachosPetersBoydMitchell2001} represent classical adaptive stencils. A $1$-ring quadric fit approximates the local $1$-ring by a paraboloid computed by least squares, in the spirit of the quadric-error reasoning of Garland and Heckbert \cite{GarlandHeckbert1997}. A normal-displacement baseline pushes the Loop midpoint along the average normal by a single scalar fit on training data, in the spirit of Phong tessellation \cite{BoubekeurAlexa2008}. A curvature-adaptive scheme uses the same gate as \PNS{} to form a $\tanh$-gated convex combination of Loop and Butterfly, but with no learned correction. An unconstrained neural baseline is a three-layer multilayer perceptron that predicts an unbounded vertex offset directly from raw $1$-ring coordinates. Neural Subdivision \cite{LiuKimChaudhuriAigermanJacobson2020}, whose relationship to \PNS{} we discussed in Section~\ref{sec:related}, occupies a different point in the design space and is treated there as related work rather than as a repeated-subdivision baseline in this paper.

We report signed-distance root-mean-square error, near-feature signed-distance root-mean-square error, Hausdorff distance, mean and maximum face-normal angular error, maximum face-normal jump after $L$ iterations, and the maximum proximity ratio $\max_i \|q_i^\theta - q_i^0\| / (C \hloc^2)$. The near-feature region is defined by $|y'| < 0.3$. All paired comparisons report bootstrap $95\%$ confidence intervals computed from $2000$ resamples, and the relative improvement of a method with mean error $E_M$ over Loop with mean error $E_{\mathrm{ref}}$ is defined as $(E_{\mathrm{ref}} - E_M)/E_{\mathrm{ref}} \times 100\%$.

\subsection{Headline metrics and repeated subdivision}
\label{sec:headline}

On the held-out ridge test set, \PNS{} improves on Loop in the metrics where $O(h^2)$ corrections are the right tool. Figure~\ref{fig:headline} shows the paired improvement of both \PNS{} and the unconstrained neural baseline. Table~\ref{tab:ridge} gives the numerical values with bootstrap confidence intervals. Three observations follow. The unconstrained baseline produces larger one-step improvements on the geometric error metrics, namely signed-distance and Hausdorff error, because free vertex prediction has more capacity to fit a target than gated-bounded correction. On normal error, however, the unconstrained baseline sharply improves near-feature accuracy but degrades the global normal error, with mean value $-4.0\%$ and lower bootstrap bound at $-16.5\%$. This pattern is a signature of high-frequency artefacts introduced by aggressive feature fitting. Finally, \PNS{} leaves the normal error essentially unchanged across the surface, with a near-feature improvement of $+0.4\%$ and a global improvement of $+0.1\%$, both statistically insignificant. Approximation gains in signed-distance and Hausdorff error are realised without compromising normal quality.

\begin{figure}[t]
\centering
\includegraphics[width=0.94\textwidth]{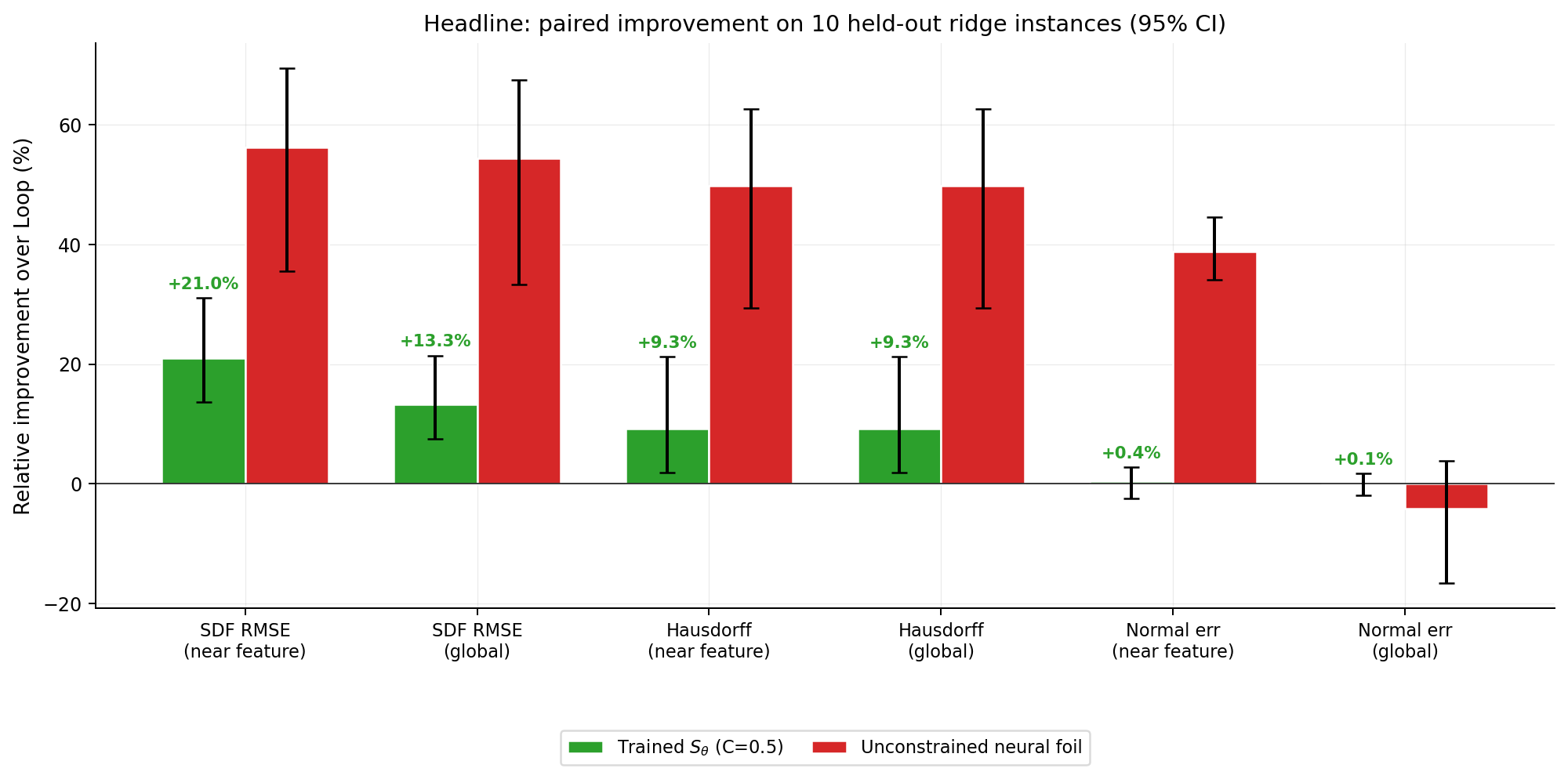}
\caption{Paired improvement over Loop on the held-out ridge test set with bootstrap $95\%$ confidence intervals at $n = 10$. The green bars correspond to trained \PNS{} with $C = 0.5$, and the red bars correspond to the unconstrained neural baseline. The unconstrained baseline achieves larger one-step improvements on geometric error metrics but degrades the global normal error. \PNS{} improves signed-distance and Hausdorff error while leaving normal error essentially unchanged.}
\label{fig:headline}
\end{figure}

\begin{table}[t]
\centering
\caption{Held-out ridge comparison with paired bootstrap $95\%$ confidence intervals at $n = 10$. Improvements are relative to Loop, and the abbreviation n.s.\ denotes a difference that is not statistically significant.}
\label{tab:ridge}
\begin{tabular*}{\textwidth}{@{\extracolsep{\fill}}l c c c c@{}}
\toprule
\textbf{Metric} & \textbf{\PNS{}} & \textbf{$95\%$ CI} & \textbf{Unconstrained} & \textbf{$95\%$ CI}\\
\midrule
SDF RMSE near feature & $+21.0\%$ & $[13.7, 31.3]$ & $+56.0\%$ & $[35.5, 69.4]$\\
SDF RMSE global & $+13.3\%$ & $[7.4, 21.4]$ & $+54.0\%$ & $[33.6, 67.5]$\\
Hausdorff near feature & $+9.3\%$ & $[1.7, 21.0]$ & $+49.6\%$ & $[29.4, 62.4]$\\
Hausdorff global & $+9.3\%$ & $[1.7, 21.0]$ & $+49.6\%$ & $[29.4, 62.4]$\\
Normal error near & $+0.4\%$ & $[-2.5, 2.5]$ (n.s.) & $+38.7\%$ & $[34.0, 44.5]$\\
Normal error global & $+0.1\%$ & $[-2.0, 3.6]$ (n.s.) & $-4.0\%$ & $[-16.5, 4.0]$\\
\bottomrule
\end{tabular*}
\end{table}

Figure~\ref{fig:money} makes the same tradeoff visible on one held-out ridge instance. The left and middle panels show the ridge cross-section. Loop systematically underfits the ridge cap because its fixed stencil averages the crest with its lower neighbours. Trained \PNS{} lifts the cap toward the analytic ridge without overshooting. The unconstrained baseline pushes the cap up more aggressively and tracks the analytic curve most closely at the peak. The right panel shows the per-vertex signed-distance error along the cross-section. In the shaded near-feature band both learned methods improve on Loop. Outside that band, the unconstrained baseline accumulates residual error at the ridge shoulders, which is the same signature Figure~\ref{fig:headline} reports as a $-4.0\%$ global normal error.

\begin{figure}[t]
\centering
\includegraphics[width=0.94\textwidth]{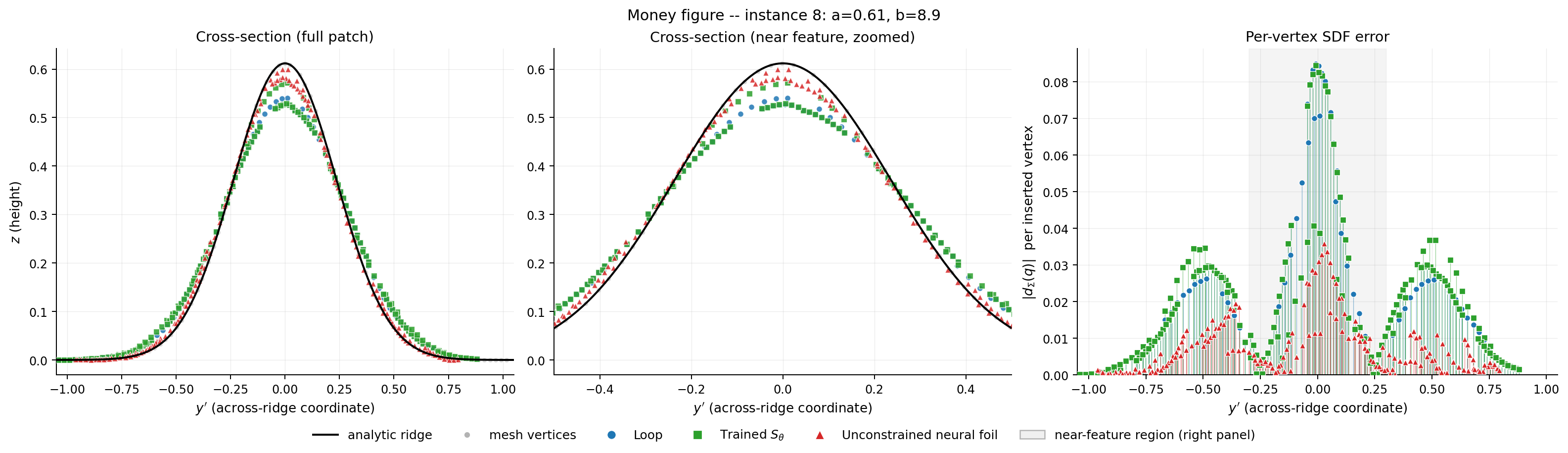}
\caption{Cross-sectional behaviour on one held-out ridge instance with parameters $a = 0.61$ and $b = 8.9$. The left panel shows the cross-section along the full patch, comparing the analytic ridge (black), Loop (blue), trained \PNS{} (green), and the unconstrained neural baseline (red). The middle panel is a zoomed view of the near-feature region. The right panel shows the per-vertex signed-distance error. In the shaded near-feature band the two learned methods improve on Loop. Outside the band, the unconstrained baseline retains visible residual error at the ridge shoulders.}
\label{fig:money}
\end{figure}

Figure~\ref{fig:repeated} reports the central empirical result of the paper. Through four levels of subdivision applied to a held-out ridge instance, three panels track approximation error, proximity ratio, and mesh regularity across methods. \PNS{} is the only method in our comparison that combines learned feature-sensitive improvement with a prescribed Loop-proximity cap under repeated subdivision. Modified Butterfly and PN-triangles are competitive at one step but exit the Loop-proximity envelope as iterations accumulate. The $1$-ring quadric and the unconstrained neural baseline leave the envelope dramatically, and the unconstrained baseline's maximum normal jump approaches $\pi$, indicating extensive face flips. The unconstrained baseline is not intended as a production method. It is a control that isolates the effect of removing the proximity envelope, and its strong one-step fit together with poor repeated behaviour illustrates the central tradeoff between expressivity and operator stability. Terminology here is deliberate. The finding does not say that Butterfly or PN-triangles are not subdivision operators. It says that under our repeated-subdivision criterion they do not remain Loop-proximate. Table~\ref{tab:repeated} gives the numerical summary on a hard ridge instance.

\begin{figure}[t]
\centering
\includegraphics[width=0.94\textwidth]{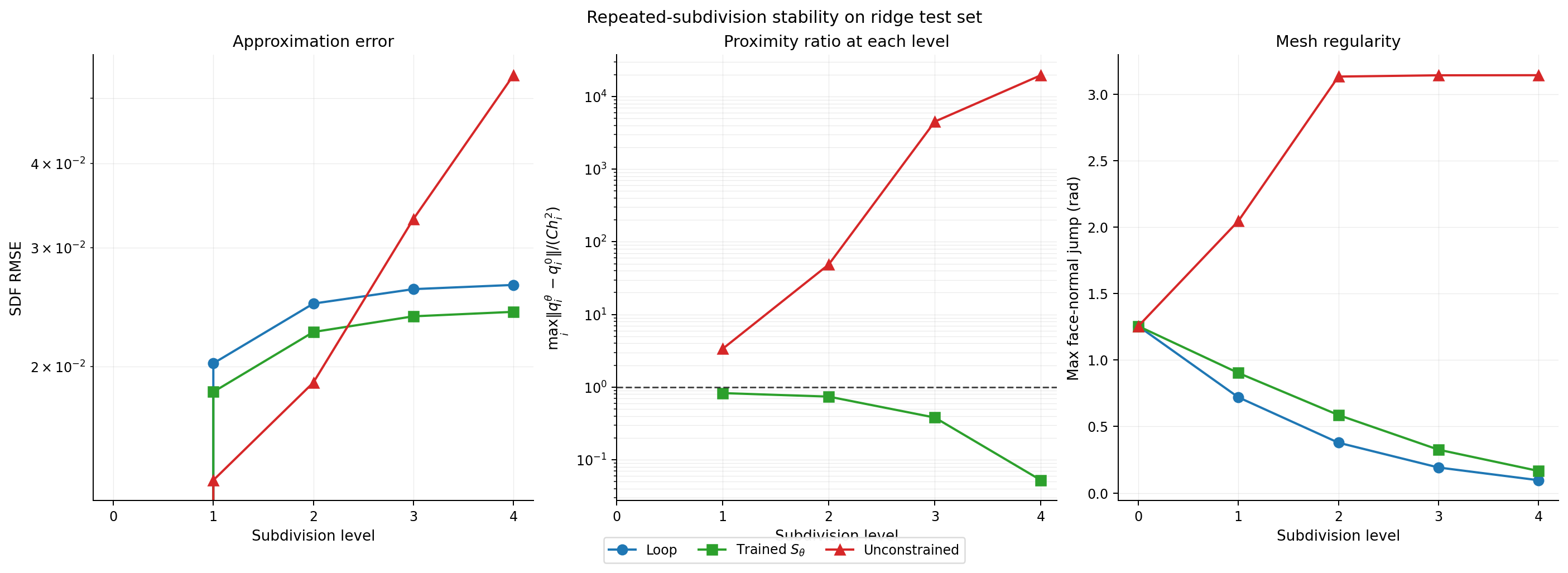}
\caption{Repeated-subdivision behaviour on the ridge test set across four refinement levels. The left panel shows signed-distance root-mean-square error against subdivision level. \PNS{} tracks Loop closely and remains bounded across all four levels, while the unconstrained baseline starts with a much lower one-step error and then diverges. The middle panel shows the maximum proximity ratio $\max_i \|q_i^\theta - q_i^0\| / (C h_i^2)$ at each level, on a logarithmic axis. \PNS{} stays at or below the architectural cap of unity throughout, while the unconstrained baseline rises past $10^4$ by level four. The right panel shows the maximum face-normal jump.}
\label{fig:repeated}
\end{figure}

\begin{table}[t]
\centering
\caption{Repeated-subdivision summary on a hard ridge instance. Loop is the reference. The last column is the maximum proximity ratio observed across four subdivision levels.}
\label{tab:repeated}
\begin{tabular*}{\textwidth}{@{\extracolsep{\fill}}l c c c c@{}}
\toprule
\textbf{Method} & \textbf{L1 max err.} & \textbf{L4 max err.} & \textbf{L4 normal jump} & \textbf{max prox.\ ratio L1--L4}\\
\midrule
Loop & 0.098 & 0.103 & 0.139 & 0.000\,(reference)\\
Butterfly & 0.092 & 0.094 & 0.862 & 61.6\\
PN-triangle & 0.076 & 0.086 & 0.858 & 8.0\\
Quadric & 0.138 & 0.210 & 3.122 & $458.7$\\
Normal-disp.\ & 0.098 & 0.103 & 0.139 & $\approx 0$\\
Curv.-adaptive & 0.097 & 0.100 & 0.390 & 5.2\\
\PNS{} & 0.098 & 0.104 & 0.130 & $\le 1$\\
Unconstrained baseline & 0.090 & 0.088 & $\approx \pi$ & $192.7$\\
\bottomrule
\end{tabular*}
\end{table}

\subsection{Capacity sweep and budget localisation}
\label{sec:budget}

Figure~\ref{fig:Csweep} shows how the approximation gain and the proximity-budget usage depend on the architectural budget $C$, swept over $0.1$, $0.25$, $0.5$, and $1.0$. Approximation gains rise from roughly $5\%$ near-feature signed-distance improvement at $C = 0.1$ to roughly $22\%$ at $C = 1.0$, with diminishing returns beyond $C = 0.5$. At every tested $C$, the maximum proximity ratio remains at or below the architectural cap to within numerical tolerance. The hyperparameter $C$ therefore behaves as a meaningful user-facing knob. Larger $C$ gives more aggressive feature fitting, while smaller $C$ tightens proximity to Loop and widens the spectral safety margin.

\begin{figure}[t]
\centering
\includegraphics[width=0.94\textwidth]{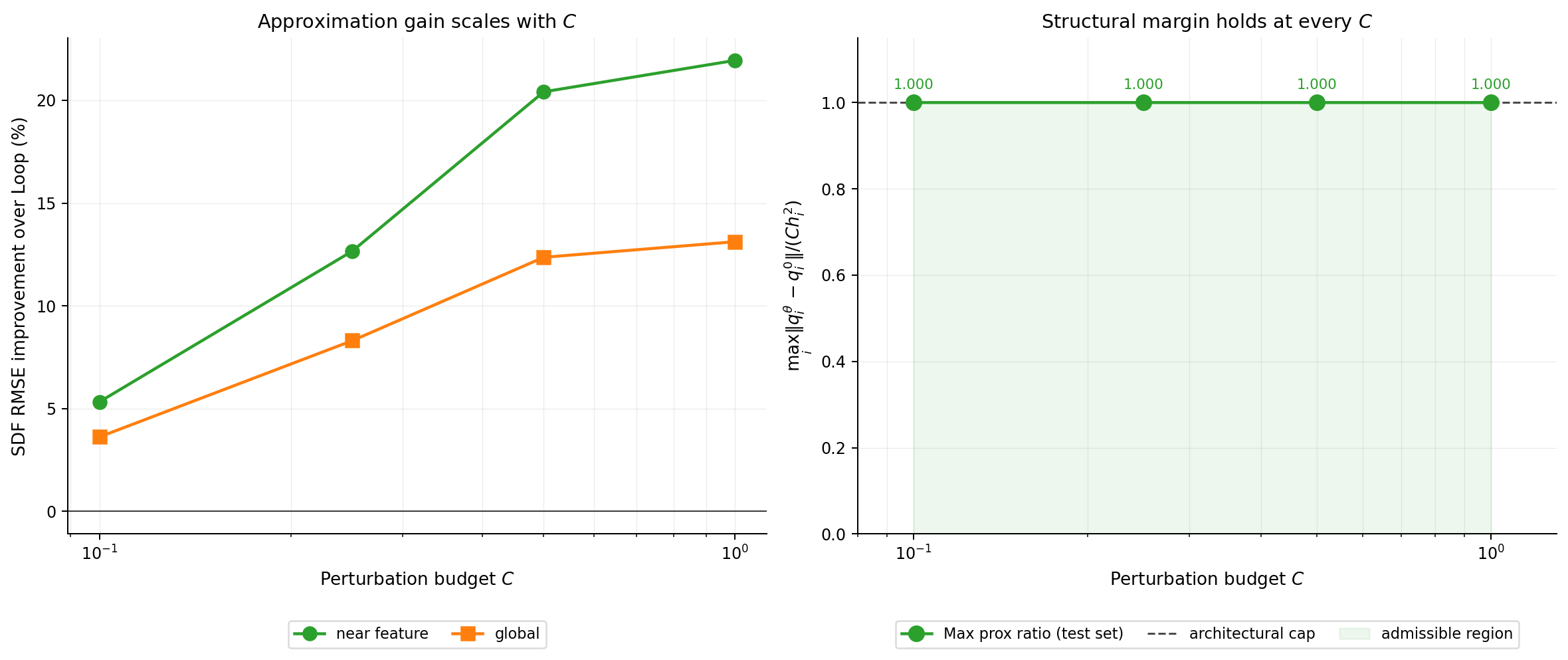}
\caption{Capacity sweep on the ridge task. The left panel shows near-feature and global signed-distance root-mean-square error improvement over Loop as a function of the perturbation budget $C$. The right panel shows that the maximum measured proximity ratio remains at the architectural cap of unity at every tested $C$, so the admissible region below the cap is fully available.}
\label{fig:Csweep}
\end{figure}

The trained operator does not spend its budget uniformly. Figure~\ref{fig:budget} shows the per-edge proximity ratio rendered on the test mesh, together with a scatter plot against the across-ridge coordinate $|y'|$. The proximity ratio tracks the gate closely and concentrates on ridge-parallel edges near the feature. Away from the ridge, the gate is close to zero and the operator reduces to Loop. Figure~\ref{fig:budget-verification} makes the mechanism explicit. The left panel plots the per-edge proximity ratio directly against the gate value $\Gate$. The points lie on the identity line, so the trained network saturates its architectural budget exactly where the gate permits it. The right panel repeats the plot from Figure~\ref{fig:budget}, now coloured by the alignment $|T_i \cdot t_{\mathrm{ridge}}|$ between the edge tangent and the ridge tangent. Both saturating clusters, at the ridge crest and at the shoulder, consist of edges parallel to the ridge, while perpendicular edges receive essentially no correction. The budget therefore concentrates on the geometrically privileged direction, which is where a fixed Loop stencil systematically underfits.

\begin{figure}[t]
\centering
\includegraphics[width=0.94\textwidth]{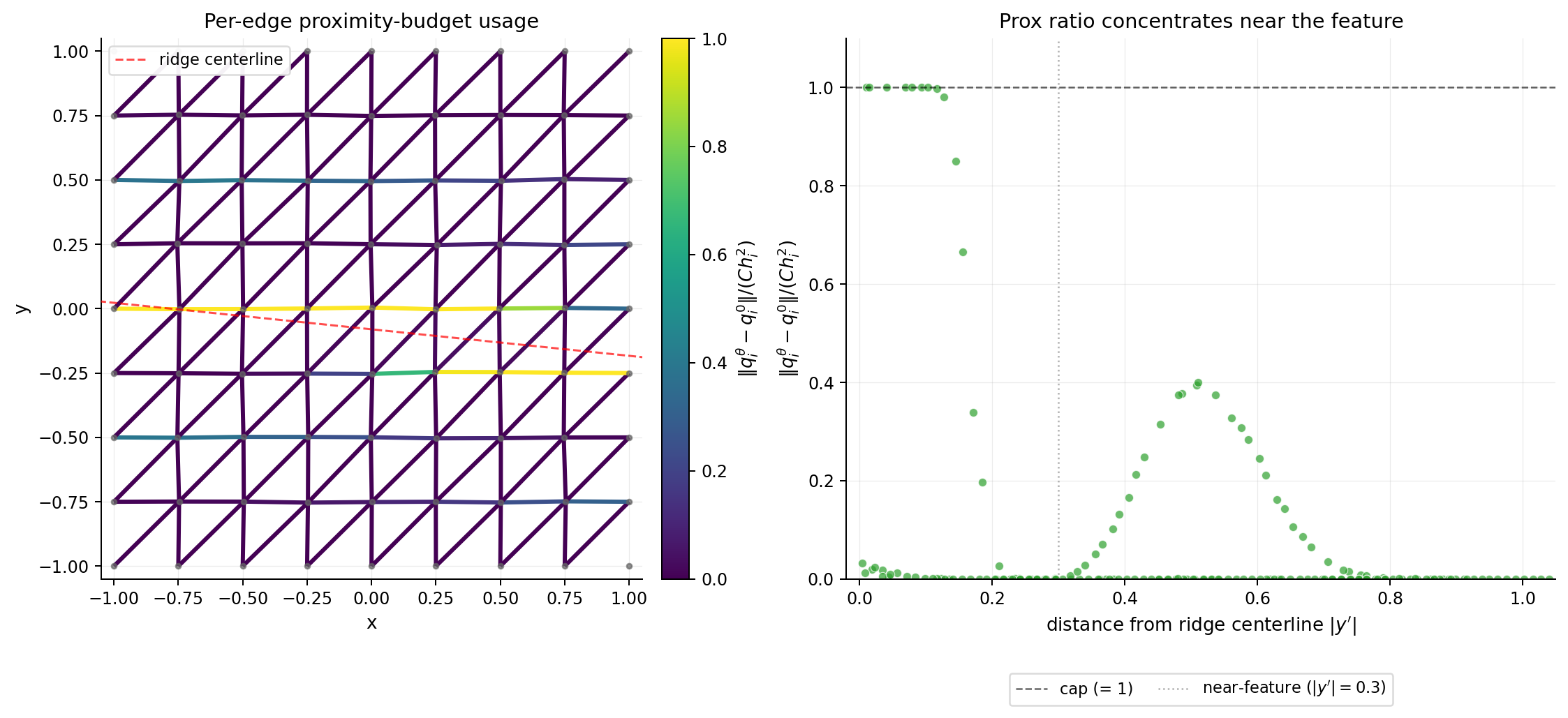}
\caption{Budget localisation. The left panel shows the per-edge proximity ratio rendered on the test mesh in a top-down view, with the ridge centreline drawn dashed. Edges parallel to the ridge crest saturate the cap, while interior-region edges use at most a few percent of the budget. The right panel shows the same per-edge proximity ratio plotted against distance from the ridge centreline. The secondary lobe near $|y'| \approx 0.5$ corresponds to the ridge shoulder, where the second derivative of the surface changes sign.}
\label{fig:budget}
\end{figure}

\begin{figure}[t]
\centering
\includegraphics[width=0.94\textwidth]{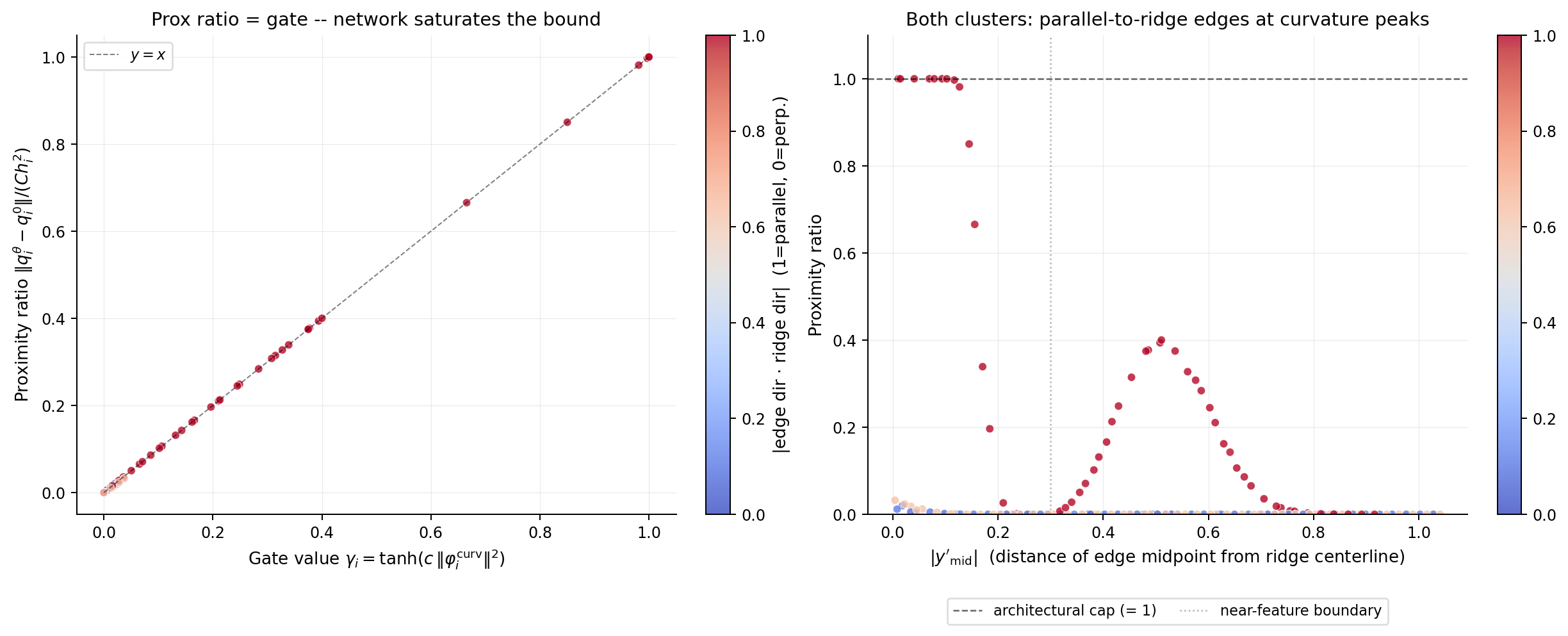}
\caption{Budget verification. The left panel shows the per-edge proximity ratio plotted against the gate value $\Gate$. The points cluster on the identity line, which demonstrates that the trained network saturates the architectural bound at the gate. The right panel shows the proximity ratio plotted against distance from the ridge centreline, coloured by $|T_i \cdot t_{\mathrm{ridge}}|$. Red corresponds to edges parallel to the ridge crest and blue to edges perpendicular to it. The two saturating clusters consist entirely of ridge-parallel edges.}
\label{fig:budget-verification}
\end{figure}

Finally, we tested a wider, deeper network with three times the parameter count and twice the training steps. The bootstrap confidence interval on the near-feature signed-distance improvement overlaps the smaller-network result, which suggests that on this task the operator class is approximation-limited by the $Ch^2$ envelope rather than by network capacity. On a smooth saddle dataset, \PNS{} provides essentially no improvement over Loop, which is the expected behaviour. Saddle surfaces are already well approximated by Loop, with truncation error that is small and not localised, so there is little for a curvature-gated correction to do and the gate stays close to zero.

\section{Discussion}
\label{sec:discussion}

\subsection{Architecture and training contributions}

The architecture provides the guarantee, while training selects a useful member of the guaranteed class. This separation is central to the method. The bound on the network output is enforced by a single normalisation. The frame is constructed analytically from the $1$-ring. The gate is a single scalar nonlinearity. None of these elements depends on a training loss term, and none can be unlearned by an adversarial weight choice. The unconstrained neural baseline routinely beats \PNS{} at one step, and this is consistent with the framing of the paper. Free vertex prediction has more capacity to fit a target than gated-bounded correction, and the same flexibility makes it not a subdivision operator. Under repeated application its iterates leave the proximity envelope, develop high-frequency artefacts, and produce flipped faces. The pathology is structural rather than incidental. A one-step regressor is not designed to be applied to its own output.

\PNS{} is best understood as finite-level feature-sensitive refinement with asymptotic structural safety. Because corrections are scaled by $\hloc^2$, their magnitude vanishes under repeated refinement on shape-regular sequences. This limits asymptotic freedom but preserves subdivision behaviour. The recent Heisenberg example of Ugail and Carriazo \cite{UgailCarriazo2026} makes clear why the architectural approach is worth the design effort. In that construction, a natural coordinate correction turns the classical four-point subdivision rule into a nonlinear scheme that is only linearly close to its Euclidean reference, and the central-channel limit is Zygmund but not $C^1$. The correction is harmless at any one refinement step, and yet its repeated injection at every scale is precisely what destroys smoothness. \PNS{} is designed to avoid the analogous failure mode by construction, keeping the correction $O(h^2)$ small at every step and letting its amplitude vanish under repeated refinement.

The hyperparameter $C$ plays three roles simultaneously, all of which point in the same direction. Larger $C$ allows more correction and serves as an approximation-capacity parameter. Smaller $C$ tightens proximity to Loop and serves as a proximity-budget parameter. Smaller $C$ also gives a wider Reif-gap margin against off-planar perturbation and serves as a spectral-safety parameter. It is the single user-facing knob in the implementation, and it has a clean interpretation across all three roles.

\subsection{Implementation and scope}
\label{sec:implementation}

The operator is implemented in PyTorch and exposed as a Blender add-on, which applies one or more steps of \PNS{} to a selected mesh. The user-facing parameter is the budget $C$, exposed as a slider. Secondary controls toggle visualisation of the per-edge gate and the proximity ratio. Render-time figures in this paper were produced through an export pipeline that takes the saved mesh assets and produces shaded, signed-distance, and proximity-heatmap renders under shared camera and shared colour management. \PNS{} requires evaluating a small multilayer perceptron per edge in addition to the Loop stencil. The network takes a fixed-size feature vector as input and produces a three-dimensional bounded output. Its evaluation is local, edge-parallel, and independent across edges, so it batches naturally and ports straightforwardly to graphics processing units. The total per-edge overhead over Loop consists of the feature construction, one small matrix-vector product per hidden layer, a single scalar nonlinearity for the gate, and the frame construction from the $1$-ring, all of which are constant-time in the mesh size. The overall cost of \PNS{} scales linearly with the number of interior edges, matching the asymptotic cost of Loop itself.

Several natural questions extend beyond the experiments reported here. On stability under further iteration, the results reach four subdivision levels, and the proximity ratio for \PNS{} decreases with level, consistent with the $\hloc^2$ scaling. There is no reason to expect qualitatively different behaviour at deeper levels, though a dedicated extended-iteration study on reduced-resolution patches, run to face-count exhaustion under a common budget, would give a stronger empirical statement. On non-planar extraordinary vertices, Theorem~\ref{thm:planar-spectral} is exact only at the planar reference, and off-planar behaviour is bounded by the perturbation magnitude and by $C$ but is neither reported nor tightened in the present experiments. A controlled study on valence-$k$ stars with a centre-lift parameter would quantify the tangent-eigenvalue drift and Reif-gap collapse that the theorem does not itself constrain. On generalisation to non-synthetic meshes, the analytic ridge benchmark isolates a controlled failure mode of fixed subdivision stencils and does so at the price of not covering the full range of practical shapes. The architectural guarantees continue to hold on any mesh, since they are properties of the operator rather than of the training data. On computer-aided design patches with closed-form ground truth, the same regime as on the ridge benchmark is expected, since the operator responds to local curvature concentration rather than to the specific analytic form. On sharp discontinuities, on noisy scans, and on broad low-curvature regions, the gate either saturates locally or suppresses the learned correction, and the method remains close to Loop by design.

\subsection{Limitations}
\label{sec:limitations}

\PNS{} is currently limited in several ways worth stating explicitly. The method delivers finite-level improvement rather than a new arbitrary limit surface, since corrections scale as $\hloc^2$ and \PNS{} becomes increasingly Loop-like under repeated refinement. Users seeking a free neural limit-surface generator outside the Loop-proximate regime will not find one here. We prove planar spectral inheritance and provide numerical off-planar certification, but not a global $C^1$ theorem at extraordinary vertices. The current evaluation emphasises localised smooth features because ridges are controlled and analytically measurable. The construction is built around Loop subdivision and triangle connectivity. Catmull--Clark subdivision and quad meshes are a natural extension, as is the $\sqrt{3}$-subdivision scheme of Kobbelt \cite{Kobbelt2000}. The same architectural recipe of bounded gated correction in a covariant frame with $h^2$ scaling should apply, but we have not implemented or tested it. Performance depends on the chosen curvature and shape descriptors, and other invariant feature sets are compatible with the architecture. Boundary handling is limited, since the current implementation routes boundary edges through Loop's boundary midpoint and does not learn a boundary correction. Non-manifold cases lie outside the standing assumptions of Section~\ref{sec:theory}. Finally, the bounded correction trades off against one-step accuracy. Users who need a one-shot reconstruction with no iteration semantics may prefer an unconstrained method.

\section{Conclusion}
\label{sec:conclusion}

We have presented Proximity-Preserving Neural Subdivision, a constrained learned refinement operator that adapts Loop subdivision while remaining inside a subdivision-theoretic proximity envelope. By embedding boundedness, equivariance, affine reproduction, and planar spectral inheritance into the architecture, the method preserves these properties for any finite weights. Training selects a useful member of this constrained operator class. The trained operator improves feature-sensitive approximation on held-out ridge surfaces and remains stable under repeated refinement, while several adaptive baselines and an unconstrained neural baseline, all competitive at one step, leave the Loop-proximity envelope as iterations accumulate. 

The broader point we like to emphasise here is that classical subdivision provides reliability without adaptivity, unconstrained neural refinement provides adaptivity without reliability, and \PNS{} provides adaptive refinement while preserving the structural regime of classical subdivision. Learning, in this view, should not replace classical structure but should be constrained to operate inside it.

\section*{Data availability}
The complete experimental code, including the subdivision scheme and results, is
openly available at \url{https://github.com/ugail/Proximity-Preserving-Neural-Subdivision}.



\begin{thebibliography}{99}

\bibitem{Loop1987}
Loop, C. T. (1987). Smooth subdivision surfaces based on triangles. Master's thesis, Department of
Mathematics, University of Utah.

\bibitem{CatmullClark1978}
Catmull, E., Clark, J. (1978). Recursively generated B-spline surfaces on arbitrary topological meshes.
\textit{Computer-Aided Design}, 10(6), 350--355.
\url{https://doi.org/10.1016/0010-4485(78)90110-0}

\bibitem{DooSabin1978}
Doo, D., Sabin, M. (1978). Behaviour of recursive division surfaces near extraordinary points.
\textit{Computer-Aided Design}, 10(6), 356--360.
\url{https://doi.org/10.1016/0010-4485(78)90111-2}

\bibitem{Reif1995}
Reif, U. (1995). A unified approach to subdivision algorithms near extraordinary vertices.
\textit{Computer Aided Geometric Design}, 12(2), 153--174.
\url{https://doi.org/10.1016/0167-8396(94)00007-F}

\bibitem{PetersReif2008}
Peters, J., Reif, U. (2008). \textit{Subdivision Surfaces}. Geometry and Computing, vol. 3, Springer.
\url{https://doi.org/10.1007/978-3-540-76406-9}

\bibitem{Hoppe1994}
Hoppe, H., DeRose, T., Duchamp, T., Halstead, M., Jin, H., McDonald, J., Schweitzer, J., Stuetzle, W.
(1994). Piecewise smooth surface reconstruction. In \textit{Proceedings of the 21st Annual Conference on
Computer Graphics and Interactive Techniques} (SIGGRAPH), 295--302.
\url{https://doi.org/10.1145/192161.192233}

\bibitem{ZorinSchroederSweldens1996}
Zorin, D., Schr\"oder, P., Sweldens, W. (1996). Interpolating subdivision for meshes with arbitrary
topology. In \textit{Proceedings of the 23rd Annual Conference on Computer Graphics and Interactive
Techniques} (SIGGRAPH), 189--192.
\url{https://doi.org/10.1145/237170.237254}

\bibitem{Kobbelt2000}
Kobbelt, L. (2000). $\sqrt{3}$-subdivision. In \textit{Proceedings of the 27th Annual Conference on
Computer Graphics and Interactive Techniques} (SIGGRAPH), 103--112.
\url{https://doi.org/10.1145/344779.344835}

\bibitem{VlachosPetersBoydMitchell2001}
Vlachos, A., Peters, J., Boyd, C., Mitchell, J. L. (2001). Curved PN triangles. In \textit{Proceedings of
the 2001 Symposium on Interactive 3D Graphics} (I3D), 159--166.
\url{https://doi.org/10.1145/364338.364387}

\bibitem{BoubekeurAlexa2008}
Boubekeur, T., Alexa, M. (2008). Phong tessellation. \textit{ACM Transactions on Graphics}, 27(5),
141:1--141:5.
\url{https://doi.org/10.1145/1409060.1409094}

\bibitem{LiuKimChaudhuriAigermanJacobson2020}
Liu, H.-T. D., Kim, V. G., Chaudhuri, S., Aigerman, N., Jacobson, A. (2020). Neural subdivision.
\textit{ACM Transactions on Graphics}, 39(4), 124:1--124:16.
\url{https://doi.org/10.1145/3386569.3392418}

\bibitem{ChenKimAigermanJacobson2023}
Chen, Y.-C., Kim, V. G., Aigerman, N., Jacobson, A. (2023). Neural progressive meshes. In \textit{ACM
SIGGRAPH 2023 Conference Proceedings}, Article 84, 1--9.
\url{https://doi.org/10.1145/3588432.3591531}

\bibitem{HanockaHertzFishGiryesFleishmanCohenOr2019}
Hanocka, R., Hertz, A., Fish, N., Giryes, R., Fleishman, S., Cohen-Or, D. (2019). MeshCNN: a network with
an edge. \textit{ACM Transactions on Graphics}, 38(4), 90:1--90:12.
\url{https://doi.org/10.1145/3306346.3322959}

\bibitem{HuLiuGuoCaiHuangMuMartin2022}
Hu, S.-M., Liu, Z.-N., Guo, M.-H., Cai, J.-X., Huang, J., Mu, T.-J., Martin, R. R. (2022).
Subdivision-based mesh convolution networks. \textit{ACM Transactions on Graphics}, 41(3), 25:1--25:16.
\url{https://doi.org/10.1145/3506694}

\bibitem{PotamiasPloumpisZafeiriou2022}
Potamias, R. A., Ploumpis, S., Zafeiriou, S. (2022). Neural mesh simplification. In \textit{IEEE/CVF
Conference on Computer Vision and Pattern Recognition} (CVPR), 18562--18571.
\url{https://doi.org/10.1109/CVPR52688.2022.01803}

\bibitem{DynGregoryLevin1987}
Dyn, N., Gregory, J. A., Levin, D. (1987). A 4-point interpolatory subdivision scheme for curve design.
\textit{Computer Aided Geometric Design}, 4(4), 257--268.
\url{https://doi.org/10.1016/0167-8396(87)90001-X}

\bibitem{Taubin1995}
Taubin, G. (1995). A signal processing approach to fair surface design. In \textit{Proceedings of the 22nd
Annual Conference on Computer Graphics and Interactive Techniques} (SIGGRAPH), 351--358.
\url{https://doi.org/10.1145/218380.218473}

\bibitem{Stam1998}
Stam, J. (1998). Exact evaluation of Catmull--Clark subdivision surfaces at arbitrary parameter values.
In \textit{Proceedings of the 25th Annual Conference on Computer Graphics and Interactive Techniques}
(SIGGRAPH), 395--404.
\url{https://doi.org/10.1145/280814.280945}

\bibitem{WarrenWeimer2002}
Warren, J., Weimer, H. (2002). \textit{Subdivision Methods for Geometric Design: A Constructive Approach}.
Morgan Kaufmann.
\url{https://doi.org/10.1016/B978-1-55860-446-9.X5000-5}

\bibitem{PrautzschBoehmPaluszny2002}
Prautzsch, H., B\"ohm, W., Paluszny, M. (2002). \textit{B'ezier and B-Spline Techniques}. Mathematics and
Visualization, Springer.
\url{https://doi.org/10.1007/978-3-662-04919-8}

\bibitem{DynLevin2002}
Dyn, N., Levin, D. (2002). Subdivision schemes in geometric modelling. \textit{Acta Numerica}, 11,
73--144.
\url{https://doi.org/10.1017/S0962492902000028}

\bibitem{BloorWilson1989}
Bloor, M. I. G., Wilson, M. J. (1989). Generating blend surfaces using partial differential equations.
\textit{Computer-Aided Design}, 21(3), 165--171.
\url{https://doi.org/10.1016/0010-4485(89)90071-7}

\bibitem{BloorWilson1990}
Bloor, M. I. G., Wilson, M. J. (1990). Using partial differential equations to generate free-form surfaces.
\textit{Computer-Aided Design}, 22(4), 202--212.
\url{https://doi.org/10.1016/0010-4485(90)90049-I}

\bibitem{UgailBloorWilson1999a}
Ugail, H., Bloor, M. I. G., Wilson, M. J. (1999). Techniques for interactive design using the PDE method.
\textit{ACM Transactions on Graphics}, 18(2), 195--212.
\url{https://doi.org/10.1145/318009.318078}

\bibitem{UgailBloorWilson1999b}
Ugail, H., Bloor, M. I. G., Wilson, M. J. (1999). Manipulation of PDE surfaces using an interactively defined parameterisation. \textit{Computers \& Graphics}, 23(4), 525--534.
\url{https://doi.org/10.1016/S0097-8493(99)00071-0}

\bibitem{KubiesaUgailWilson2004}
Kubiesa, S., Ugail, H., Wilson, M. J. (2004). Interactive design using higher order PDEs.
\textit{The Visual Computer}, 20(10), 682--693.
\url{https://doi.org/10.1007/s00371-004-0261-3}

\bibitem{AthanasopoulosUgailCastro2009}
Athanasopoulos, M., Ugail, H., Gonz\'alez Castro, G. (2009). Parametric design of aircraft geometry using
partial differential equations. \textit{Advances in Engineering Software}, 40(7), 479--486.
\url{https://doi.org/10.1016/j.advengsoft.2008.08.001}

\bibitem{ShengSourinCastroUgail2010}
Sheng, Y., Sourin, A., Gonz\'alez Castro, G., Ugail, H. (2010). A PDE method for patchwise approximation
of large polygon meshes. \textit{The Visual Computer}, 26(6--8), 975--984.
\url{https://doi.org/10.1007/s00371-010-0456-8}

\bibitem{ShengWillisCastroUgail2011}
Sheng, Y., Willis, P., Gonz\'alez Castro, G., Ugail, H. (2011). Facial geometry parameterisation based on
partial differential equations. \textit{Mathematical and Computer Modelling}, 54(5--6), 1536--1548.
\url{https://doi.org/10.1016/j.mcm.2011.04.025}

\bibitem{MonterdeUgail2004}
Monterde, J., Ugail, H. (2004). On harmonic and biharmonic B'ezier surfaces.
\textit{Computer Aided Geometric Design}, 21(7), 697--715.
\url{https://doi.org/10.1016/j.cagd.2004.07.003}

\bibitem{MonterdeUgail2006}
Monterde, J., Ugail, H. (2006). A general 4th-order PDE method to generate B'ezier surfaces from the
boundary. \textit{Computer Aided Geometric Design}, 23(2), 208--225.
\url{https://doi.org/10.1016/j.cagd.2005.09.001}

\bibitem{GonzalezCastroUgailWillisPalmer2008}
Gonz\'alez Castro, G., Ugail, H., Willis, P., Palmer, I. (2008). A survey of partial differential equations
in geometric design. \textit{The Visual Computer}, 24(3), 213--225.
\url{https://doi.org/10.1007/s00371-007-0190-z}

\bibitem{WallnerDyn2005}
Wallner, J., Dyn, N. (2005). Convergence and $C^1$ analysis of subdivision schemes on manifolds by
proximity. \textit{Computer Aided Geometric Design}, 22(7), 593--622.
\url{https://doi.org/10.1016/j.cagd.2005.06.003}

\bibitem{Grohs2009}
Grohs, P. (2008). Smoothness analysis of subdivision schemes on regular grids by proximity.
\textit{SIAM Journal on Numerical Analysis}, 46(4), 2169--2182.
\url{https://doi.org/10.1137/060669759}

\bibitem{Grohs2010}
Grohs, P. (2010). Approximation order from stability of nonlinear subdivision schemes.
\textit{Journal of Approximation Theory}, 162(5), 1085--1094.
\url{https://doi.org/10.1016/j.jat.2009.12.003}

\bibitem{HueningWallner2019}
H\"uning, S., Wallner, J. (2019). Convergence of subdivision schemes on Riemannian manifolds with
nonpositive sectional curvature. \textit{Advances in Computational Mathematics}, 45(3), 1689--1709.
\url{https://doi.org/10.1007/s10444-019-09693-x}

\bibitem{HueningWallner2022}
H\"uning, S., Wallner, J. (2022). Convergence analysis of subdivision processes on the sphere.
\textit{IMA Journal of Numerical Analysis}, 42(1), 698--711.
\url{https://doi.org/10.1093/imanum/draa086}

\bibitem{UgailCarriazo2026}
Ugail, H., Carriazo, A. (2026). A Heisenberg subdivision scheme with central smoothness loss.
arXiv:2607.05446.
\url{https://doi.org/10.48550/arXiv.2607.05446}

\bibitem{Hoppe1996}
Hoppe, H. (1996). Progressive meshes. In \textit{Proceedings of the 23rd Annual Conference on Computer
Graphics and Interactive Techniques} (SIGGRAPH), 99--108.
\url{https://doi.org/10.1145/237170.237216}

\bibitem{MasciBoscainiBronsteinVandergheynst2015}
Masci, J., Boscaini, D., Bronstein, M. M., Vandergheynst, P. (2015). Geodesic convolutional neural
networks on Riemannian manifolds. In \textit{IEEE International Conference on Computer Vision Workshops}
(ICCV Workshops), 832--840.
\url{https://doi.org/10.1109/ICCVW.2015.112}

\bibitem{CohenWeilerKicanaogluWelling2019}
Cohen, T. S., Weiler, M., Kicanaoglu, B., Welling, M. (2019). Gauge equivariant convolutional networks
and the icosahedral CNN. In \textit{Proceedings of the 36th International Conference on Machine Learning}
(ICML), Proceedings of Machine Learning Research, vol.\ 97, 1321--1330.

\bibitem{deHaanWeilerCohenWelling2021}
de Haan, P., Weiler, M., Cohen, T., Welling, M. (2021). Gauge equivariant mesh CNNs: anisotropic
convolutions on geometric graphs. In \textit{International Conference on Learning Representations}
(ICLR).

\bibitem{FeyLenssenWeichertMueller2018}
Fey, M., Lenssen, J. E., Weichert, F., M\"uller, H. (2018). SplineCNN: fast geometric deep learning with
continuous B-spline kernels. In \textit{IEEE/CVF Conference on Computer Vision and Pattern Recognition}
(CVPR), 869--877.
\url{https://doi.org/10.1109/CVPR.2018.00097}

\bibitem{QiSuMoGuibas2017}
Qi, C. R., Su, H., Mo, K., Guibas, L. J. (2017). PointNet: deep learning on point sets for 3D
classification and segmentation. In \textit{IEEE Conference on Computer Vision and Pattern Recognition}
(CVPR), 77--85.
\url{https://doi.org/10.1109/CVPR.2017.16}

\bibitem{QiYiSuGuibas2017}
Qi, C. R., Yi, L., Su, H., Guibas, L. J. (2017). PointNet++: deep hierarchical feature learning on point
sets in a metric space. In \textit{Advances in Neural Information Processing Systems} (NeurIPS), vol.\ 30,
5099--5108.

\bibitem{WangSunLiuSarmaBronsteinSolomon2019}
Wang, Y., Sun, Y., Liu, Z., Sarma, S. E., Bronstein, M. M., Solomon, J. M. (2019). Dynamic graph CNN for
learning on point clouds. \textit{ACM Transactions on Graphics}, 38(5), 146:1--146:12.
\url{https://doi.org/10.1145/3326362}

\bibitem{DefferrardBressonVandergheynst2016}
Defferrard, M., Bresson, X., Vandergheynst, P. (2016). Convolutional neural networks on graphs with fast
localized spectral filtering. In \textit{Advances in Neural Information Processing Systems} (NeurIPS),
vol.\ 29, 3844--3852.

\bibitem{BotschKobbeltPaulyAlliezLevy2010}
Botsch, M., Kobbelt, L., Pauly, M., Alliez, P., L\'evy, B. (2010). \textit{Polygon Mesh Processing}.
A K Peters / CRC Press.
\url{https://doi.org/10.1201/b10688}

\bibitem{Meyer2003}
Meyer, M., Desbrun, M., Schr\"oder, P., Barr, A. H. (2003). Discrete differential-geometry operators for
triangulated 2-manifolds. In \textit{Visualization and Mathematics III}, Hege, H.-C., Polthier, K.\ (Eds.),
Springer, 35--57.
\url{https://doi.org/10.1007/978-3-662-05105-4_2}

\bibitem{Desbrun2000}
Desbrun, M., Meyer, M., Schr\"oder, P., Barr, A. H. (2000). Discrete differential-geometry operators in nD.
Technical report, California Institute of Technology.

\bibitem{KingmaBa2015}
Kingma, D. P., Ba, J. (2015). Adam: a method for stochastic optimization. In \textit{International
Conference on Learning Representations} (ICLR).

\bibitem{GarlandHeckbert1997}
Garland, M., Heckbert, P. S. (1997). Surface simplification using quadric error metrics. In
\textit{Proceedings of the 24th Annual Conference on Computer Graphics and Interactive Techniques}
(SIGGRAPH), 209--216.
\url{https://doi.org/10.1145/258734.258849}

\end{thebibliography}
\end{document}